\documentclass[a4paper]{article}
\usepackage[pdftex]{graphicx}
\usepackage{amsmath}
\usepackage{amssymb}
\usepackage{amsthm}
\usepackage[utf8]{inputenc}
\usepackage[T1]{fontenc}
\usepackage{hyperref}
\usepackage{smartref}
\usepackage{todonotes}
\usepackage{fullpage}
\usepackage{epstopdf}
\usepackage{subcaption}

\newtheorem{lemma}{Lemma}
\newtheorem{theorem}{Theorem}

\usepackage[toc,page]{appendix}

\usepackage{color}

\renewcommand\vec{\boldsymbol}

\newcommand{\scal}[2]{\langle #1,#2\rangle}

\newcount\colveccount
\newcommand*\colvec[1]{
        \global\colveccount#1
        \begin{pmatrix}
        \colvecnext
}
\def\colvecnext#1{
        #1
        \global\advance\colveccount-1
        \ifnum\colveccount>0
                \\
                \expandafter\colvecnext
        \else
                \end{pmatrix}
        \fi
}

\begin{document}

\title{Beyond polarities: Generalizing Maxwell's reciprocal diagrams}
\author{Tamás Baranyai}
\maketitle

\section{Preface}
This is an updated version of the paper, containing the essential parts of the work.  As the non essential parts were mostly superficial and simplistic, they are omitted. 

\section{Notation and preliminaries}
\subsection{Basics}
We will use a linear algebraic description of projective transformations and dualities, which relies on the following factorisation. Consider vector-space $\mathbb{V}$ over the real numbers. The vectors $x,y\in \mathbb{V}$ will be called equivalent if $x=\alpha y$ for any $\alpha \in \mathbb{R}\setminus \{ 0 \} $. This will be denoted with $x\sim y$. Since this relation is an equivalence relation, it can be used to factorize vector spaces and attain what is called homogeneous coordinates \cite{pottmann2001computational}.

The three dimensional projective space ($PG(3)$) can be thought of as the union of Euclidean space and a plane at infinity. To describe points of $PG(3)$ we will use the following convention. If a point is in the Euclidean part with coordinates $\vec{a}\in \mathbb{R}^3$ then it is represented by the equivalence class of all vectors of form $(\lambda \vec{a},\lambda)$, while points at infinity are represented by equivalence classes of form $(\lambda \vec{a},0)$. (They will be treated as column vectors in the following.) Planes of $PG(3)$ are also represented with 4-vectors, in such way, that plane $p=(\vec{p},p_0)$ and point $a=(\vec{a},a_0)$ are incident (the plane contains the point) if and only if $\left\langle p,a \right\rangle=0$ (the usual scalar product in $\mathbb{R}^4$). We will deal with two types of transformations of these elements: projective transformations and dualities.

Projective transformations map points to points (planes to planes) preserving incidence and are describable with invertible matrices acting on vector-equivalence classes, such that transformation $P$ maps point $a$ to $a'=Pa$ and plane $p$ to $p'=P^{-T}p$. Note, that point $Pa$ and plane $P^{-T}p$ are incident if and only if $\scal{Pa}{P^{-T}p}=0=\scal{a}{p}$ holds, which is the incidence condition before the transformation. As the linear nature of the description suggests, we have equivalence classes of matrices representing transformations and there is a bijection between all possible projective transformations and all equivalence classes containing invertible matrices.

Dualities (in 3d) map points to planes, planes to points and lines to lines, such that the incidence relations are preserved (e.g. the dual planes of coplanar points intersect in a single point). They can also be described with the same equivalence classes of invertible matrices, but point $a$ is mapped to plane $\vec{D}a$ and plane $p$ is mapped to point $\vec{D}^{-T}p$. (The bold font is to distinguish between them and projective transformations.) Polarities are dualities which have periodicity of two, that is they satisfy  $\vec{D}^{-T}\vec{D}\sim Id$ (the 4 dimensional identity matrix).

\subsection{Projection ("camera") matrices}
The action of any projection from $PG(3)$ to $PG(2)$ can be described \cite{hartley2003multiple} with a $3\times4$ matrix (equivalence class) of rank 3, acting on the homogeneous coordinates of points. These are also called "camera" matrices, since they describe how a camera sees the world. In general, any $3\times4$ matrix of rank 3 is a camera matrix, but some of them correspond to cameras with "built in image processing". This happens for instance when pixels of a digital camera are not square, resulting in an additional distortion (affine transformation is in this case) of the image. Since we are interested in "theoretical" projections, we will restrict ourselves to a subset of projection matrices as follows.

\subsubsection{Finite pinhole cameras}
The basic camera, projecting through point $(\vec{0},0)$ onto plane $(0,0,\phi^{-1},1)$ is
\begin{align}
C_{\phi}=
\begin{bmatrix}
1 & 0 & 0 & 0\\
0 & 1 & 0 & 0\\
0 & 0 & \phi^{-1} & 0
\end{bmatrix}
\end{align}
where $\phi$ is called the focal distance. Any projection through a finite point to a finite plane can be given by applying a rotation and translation before projecting with $C_{\phi}$, that is the general form of our finite pinhole cameras will be
\begin{align}
C_f=C_{\phi}
\begin{bmatrix}
R_3 & \vec{t}\\
\vec{0}^T & 1
\end{bmatrix}\label{eq:kamera_veges}
\end{align}  
where $R_3$ is a 3 dimensional rotation matrix and $\vec{t}$ is the vector describing the translation.
\subsubsection{Parallel projections}
Similarly, the basic parallel projection along the $(0,0,1)$ axis is given by
\begin{align}
C_{\infty}=
\begin{bmatrix}
1 & 0 & 0 & 0\\
0 & 1 & 0 & 0\\
0 & 0 & 0 & 1
\end{bmatrix}.
\end{align}
while in case of finite cameras there is always a projecting ray orthogonal to the image plane, when projecting from points at infinity this may not be the case. The skewness of the image plane can be addressed with an additional stretching of the image in one direction with an axial affinity. Since the direction and location of the axis can be influenced by the choice of $R_3$ and $\vec{t}$, the general form of our parallel projections will be
\begin{align}
C_p=\begin{bmatrix}
1 & 0 & 0\\
0 & \psi & 0\\
0 & 0 & 1
\end{bmatrix} C_{\infty}
\begin{bmatrix}
R_3 & \vec{t}\\
\vec{0}^T & 1
\end{bmatrix}\label{eq:kamera_vegtelen}
\end{align}
with $\psi \in \mathbb{R} \setminus \{ 0 \}$ (here one component of $\vec{t}$ is cancelled out by the projection).

\subsubsection{Applicable projections}
Let $\mathcal{C}$ denote the set of matrices satisfying either \eqref{eq:kamera_veges} or \eqref{eq:kamera_vegtelen}. This set contains all the projection matrices allowed for our purpose.

\section{The generalized system}
Let us have polyhedron $\mathcal{P}$ such that its projection $\Gamma_1=C_1(\mathcal{P})$ on plane $j_1$ through point $i_1$ is the first diagram. Let $\mathcal{D}$ be the dual polyhedron of $\mathcal{P}$, according to some duality. The second diagram $\Gamma_2=C_2(\mathcal{D})$ is attained by projecting $\mathcal{D}$ from point $i_2$, to plane $j_2$. 

An example can be seen in Figure \ref{fig:mkep}, the whole process with homogenous coordinates is presented in Appendix \ref{sec:appa}. 
\begin{figure}[h]
\includegraphics[width=0.5\textwidth]{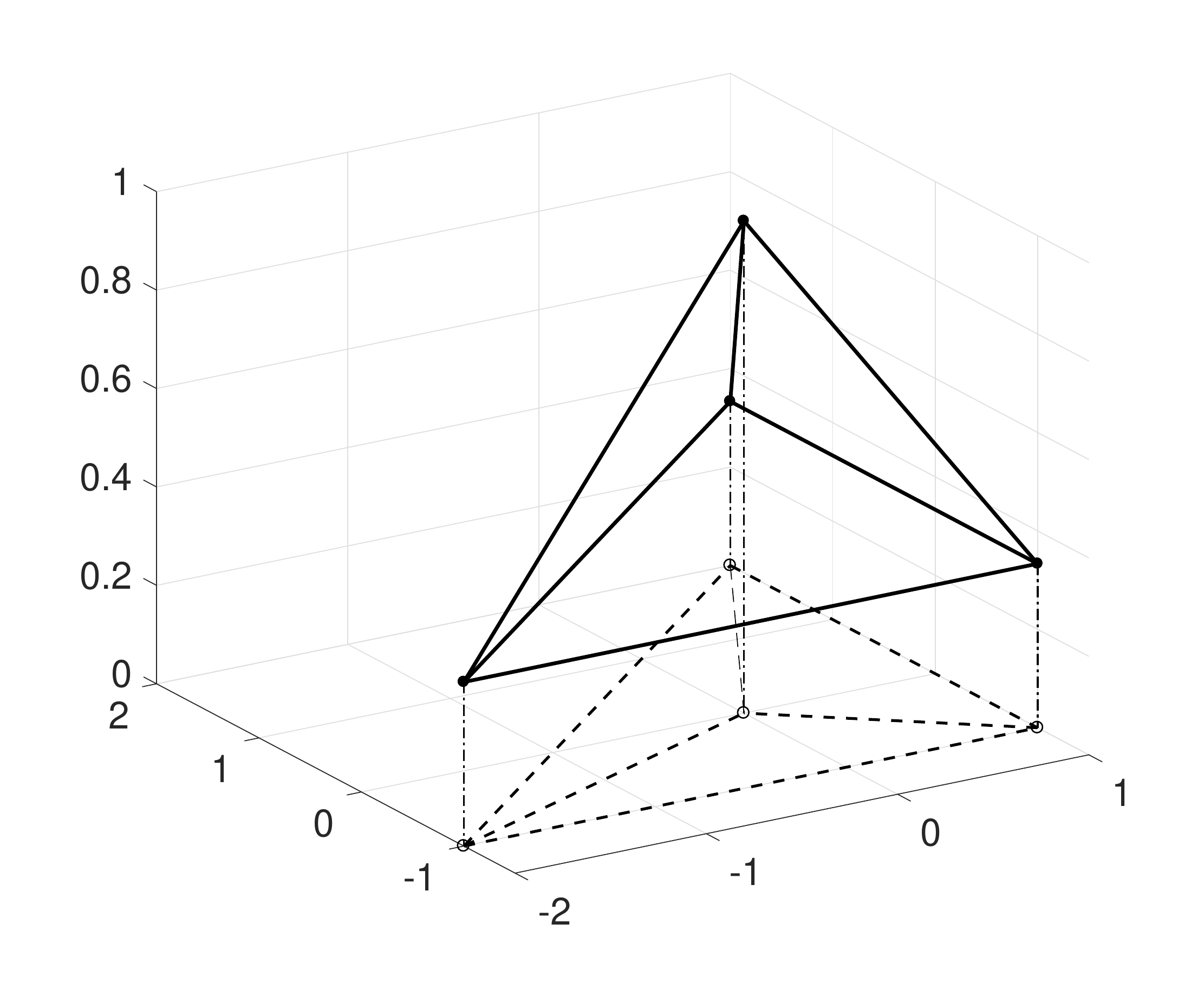}
\includegraphics[width=0.5\textwidth]{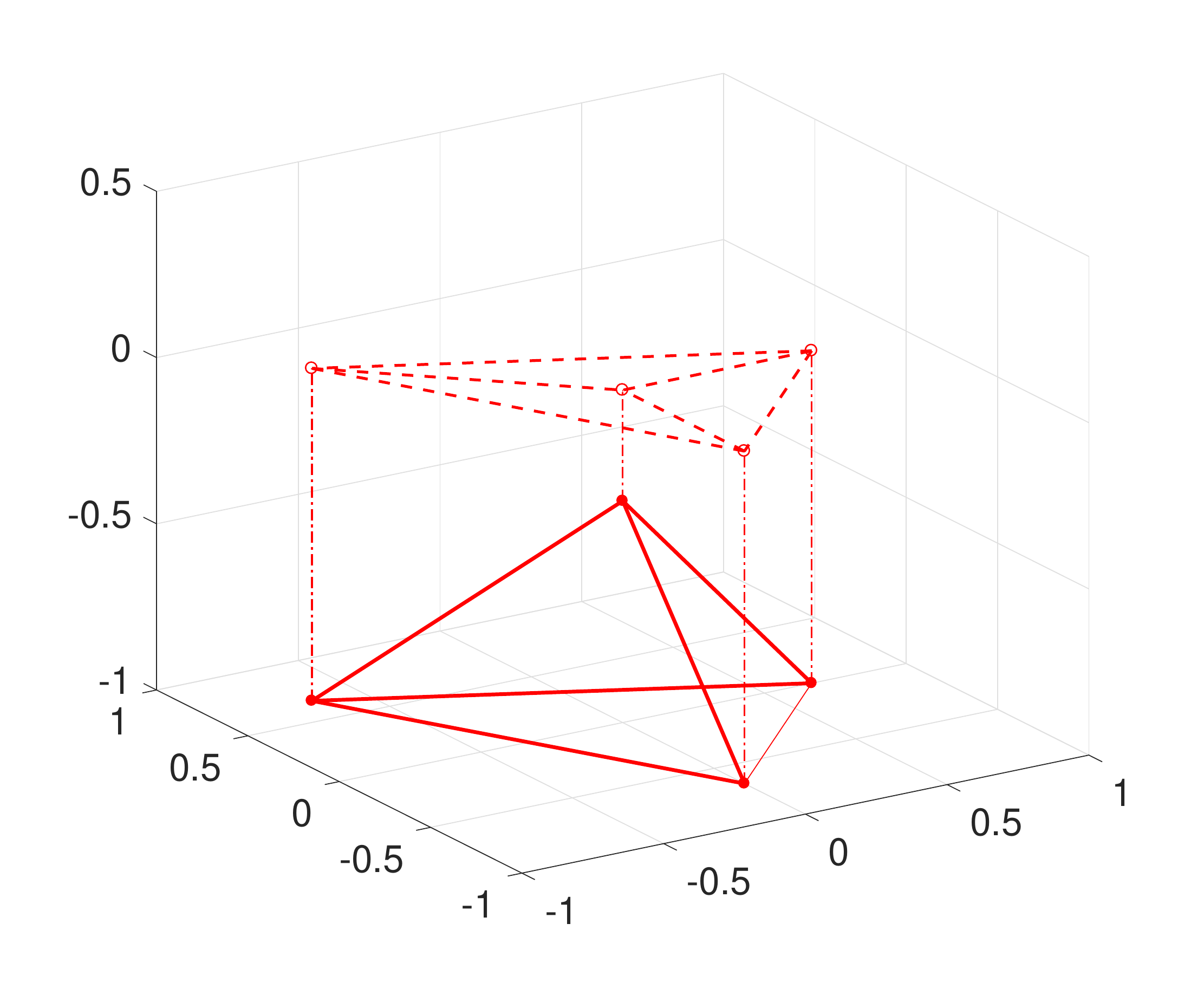}
\caption{Left: polyhedron $\mathcal{P}$ (solid lines), and its projecion. Right: polyhedron $\mathcal{D}$ (solid lines), and its projecion.}\label{fig:mkep}
\end{figure}

For the purpose of finding out whether a given duality can be used to get force diagram for a form diagram, we have to account for the fact that in graphical statics the user has freedom in choosing the precise location, scale and orientation of the images on the paper. In order to do this, we will accept solutions which are attainable from another solution through similarity transformations (amongst other things, this implies  no distinction between the so called "Maxwell-reciprocals" and "Cremona-reciprocals").  Since all projective transformations of $PG(2)$ are describable with invertible $3\times3$ matrices, for future reference let $\mathcal{S}$ denote the set of all matrices describing a similarity transformation of $PG(2)$.

When investigating the periodicity of the methods, the different iterations of the respective diagrams will have to be point-wise identical, not just similar to each other.  

\begin{theorem}\label{thm:jo}
Any projective duality can be used to get force and form diagrams and any of them has infinitely many pairs of projections to do so.
\end{theorem}
Before we go on and prove this, we will need the following two lemmas:

\begin{lemma}\label{lem:inf}
There are infinitely many pairs of applicable projections with which the canonical duality (represented by the 4 dimensional identity matrix) can be used to give form and force diagrams.
\end{lemma}
\begin{proof}
In order to show this, we will reduce the problem to a method which is known to work. One usage of Cremona's method works by having Polyhedra $\mathcal{P}$ and $\mathcal{D}$ such that points of $\mathcal{P}$ are mapped to planes of $\mathcal{D}$ by duality (null-polarity) 
\begin{align}
\vec{B}=
\begin{bmatrix}
0 & -1 & 0 & 0 \\
1 & 0 & 0 & 0 \\
0 & 0 & 0 & -1\\
0 & 0 & 1 & 0
\end{bmatrix}
\end{align}
 while the two diagrams are 
 \begin{align}
 \Gamma_1&=C_{\infty}(\mathcal{P})\\
 \Gamma_2&=C_{\infty}(\mathcal{D}).
\end{align} 
Let us note, that 
\begin{align}
e^{\pi/2\vec{B}}=\vec{B}
\end{align}
and
\begin{align}
e^{\pi\vec{B}}=-\vec{Id}\sim\vec{Id}
\end{align}
(the usual matrix exponential). Using this we can decompose the canonical duality as
\begin{align}
\vec{Id}=e^{\alpha\vec{B}} \vec{B} e^{(\pi/2-\alpha)\vec{B}}\label{eq:decompose}
\end{align}
with $\alpha\in \mathbb{R}$ (mod $\pi$). Let polyhedra $\mathcal{P}^*$ and $\mathcal{D}^*$ be defined as
\begin{align}
\mathcal{P}^*:&=e^{-\alpha\vec{B}}(\mathcal{P})\\
\mathcal{D}^*:&=e^{(\pi/2-\alpha)\vec{B}}(\mathcal{D})
\end{align} 
(note that $(e^{(pi/2-\alpha)\vec{B}})^{-T}=e^{(pi/2-\alpha)\vec{B}}$). It can be seen from \eqref{eq:decompose} that points of $\mathcal{P}^*$ are mapped to planes of $\mathcal{D}^*$ by the canonical duality. We can define $C_{\alpha}$ and $C_{\pi/2-\alpha}$ as
\begin{align}
C_{\alpha}&:=C_\infty e^{\alpha\vec{B}}\\
C_{(\alpha-\pi/2)}:&=C_\infty e^{(\alpha-\pi/2)\vec{B}}.
\end{align}
Note that $C_{\alpha}(\mathcal{P}^*)=C_\infty(\mathcal{P})$ and $C_{(\alpha-\pi/2)}(\mathcal{D}^*)=C_\infty(\mathcal{D})$ by definition. It is trivial to show that $C_\alpha\in\mathcal{C}$ for all $\alpha \in \mathbb{R} \ (mod \ \pi)$. The fact that there are infinitely many numbers in $\mathbb{R} \ (mod \ \pi)$ completes the proof.
\end{proof}

Also, this gives exactly the same images, the infinite number of solutions do not rely on similarity transformations in $\mathcal{S}$. Furthermore, $\alpha$ changes the internal parameters of the cameras (focal distance) as well as the centre of projection. 

\begin{lemma}\label{lem:maxwell}
There are infinitely many pairs of applicable projections with which the duality used by Maxwell can be used to give form and force diagrams.
\end{lemma}
\begin{proof}
Recall (or see Appendix \ref{sec:appa}) how Maxwell's method works by having Polyhedra $\mathcal{P}$ and $\mathcal{D}$ such that points of $\mathcal{P}$ are mapped to planes of $\mathcal{D}$ by duality 
\begin{align}
\vec{A}=
\begin{bmatrix}
1 & 0 & 0 & 0 \\
0 & 1 & 0 & 0 \\
0 & 0 & 0 & 1\\
0 & 0 & 1 & 0
\end{bmatrix}
\end{align}
 while the two diagrams are 
 \begin{align}
 \Gamma_1=C_{\infty}(\mathcal{P})\\
 \Gamma_2=C_{\infty}(\mathcal{D}).
\end{align} 

We can use $C_{\alpha}$ and $\mathcal{P}^*$ from the previous proof to get a family of projections and corresponding polyhedra. Let us define $\mathcal{D}'$ and $C_{\alpha'}$ as
\begin{align}
\mathcal{D}':=Ae^{\alpha\vec{B}}A(\mathcal{D})\\
C_{\alpha'}:=Ae^{-\alpha\vec{B}}A
\end{align}
after which it is not hard to see that 
\begin{align}
C_{\alpha'}(\mathcal{D}')=C_{\infty}(\mathcal{D})
\end{align}
and the fact that $C_{\alpha'}\in \mathcal{C}$ completes the proof.
\end{proof}

\begin{proof}[Proof of Theorem \ref{thm:jo}]
Recall how any matrix has an $LPU$ decomposition \cite{carrell2017groups} where $L$ and $U$ are lower and upper triangular matrices respectively, and $P$ is a permutation matrix.
We can have an analogous decomposition of duality $\vec{D}$ as
\begin{align}
\vec{D}=V^T\vec{P}U\label{eq:dc}
\end{align}
with the following 3 cases
\begin{itemize}
\item[a)] $\vec{P}=\vec{Id}$
\item[b)] $\vec{P}$ describes a permutation of axes in $\mathbb{R}^3$
\item[c)] $\vec{P}=\vec{A}$ with some additional permutation of axes in $\mathbb{R}^3$ if necessary.
\end{itemize}
\paragraph{In case a)} using the notation of Lemma \ref{lem:inf} we are looking for a family of projections $C_{\beta}$ and $C_{\gamma}$  satisfying $C_{\alpha}(\mathcal{P^*})=C_{\beta}(U^{-1}(\mathcal{P^*}))$ and $C_{(\alpha-\pi/2)}(\mathcal{D^*})=C_{\gamma}(V^{-1}(\mathcal{D^*}))$. Note, how $U^{-1}$ and $V^{-1}$ are affine transformations due to their upper triangular nature. (The transposition of the lower triangular matrix is necessary due to the fact that it describes an operation on plane-coordinates.) Now let us denote the projection centre of $C_{\alpha}$ with $i$ and the image plane of it with $k$. We can take the affine transform of the entire system, resulting in projecting $U^{-1}(\mathcal{P^*})$ through point $U^{-1}i$ to image plane $U^Tk$ (in homogeneous coordinates). Since the restriction of an affine transformation to its subspaces is also an affine transformation, what we see on $U^Tk$ is an affine image of  $C_{\alpha}(\mathcal{P^*})$. Any affine transformation of $PG(2)$ can be decomposed into a similarity transformation and an axial affinity. Since we accept similar images as solutions, we only have to undo the effect of the axial affinity, which can be done by introducing a rotated image plane, intersecting $U^Tk$ in the axis of affinity. With this we have shown the existence of the family $C_{\beta}\in\mathcal{C}$satisfying $C_{\alpha}(\mathcal{P^*})=C_{\beta}(U^{-1}(\mathcal{P^*}))$. The existence of $C_{\gamma}\in\mathcal{C}$ satisfying $C_{(\alpha-\pi/2)}(\mathcal{D^*})=C_{\gamma}(V^{-1}(\mathcal{D^*}))$ can be seen through similar reasoning. 
\paragraph{In case b)} the fact that we have to add an additional relabelling of the axes to account for the permutation of axes in in $\mathbb{R}^3$ does not change the validity of the proof.
\paragraph{In case c)} we have to rely on the families $C_{\alpha}$ and $C_{\alpha'}$, whose existence is shown in Lemma \ref{lem:maxwell}, otherwise the proof is the same.
\end{proof} 
This concludes the proof of Theorem \ref{thm:jo}. An illustration can be seen in Figure \ref{fig:affin}. The underlying numerics is presented in Appendix \ref{sec:appc}.

\begin{figure}[h]
\includegraphics[width=0.5\textwidth]{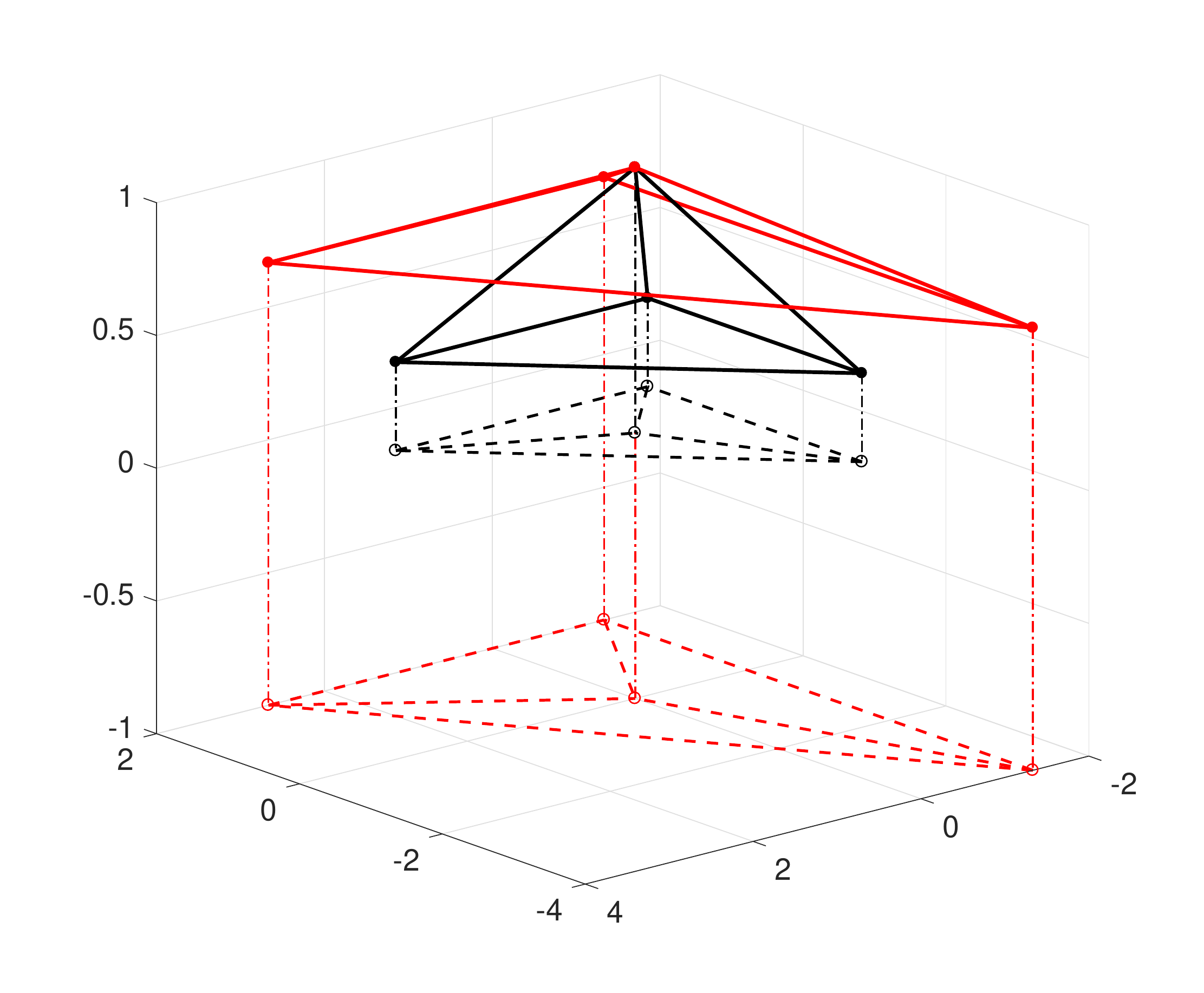}
\includegraphics[width=0.5\textwidth]{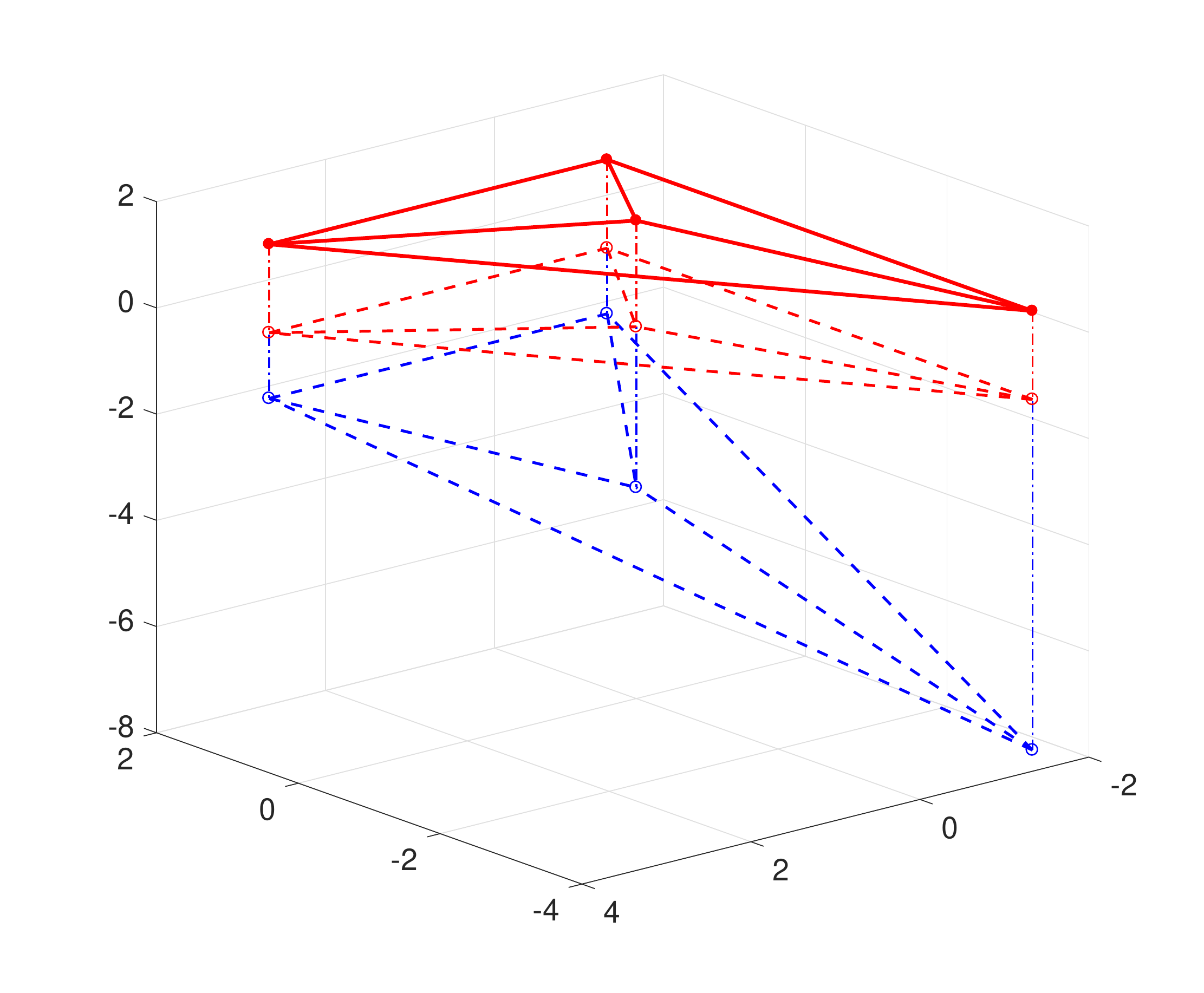}
\caption{Left: polyhedron $\mathcal{P}$ (black),  $U^{-1}(\mathcal{P})$ (red) as well as their projecions. Right: polyhedron $U^{-1}(\mathcal{P})$, its projection (red), and the corrected image on a skew plane (blue).The corrected (blue) projection of $U^{-1}(\mathcal{P})$ is a uniformly scaled up version of the starting (black) projection of $\mathcal{P}$.}\label{fig:affin}
\end{figure}

\section{Periodicity}
Maxwell called the two diagrams reciprocal, because for any one of them the method he has given produced the other diagram. If we want to generalize this, we need a method which consists of a single camera matrix and a duality which is always iterated forwards. (As opposed to the construction given in the proof of Theorem \ref{thm:jo}, where we might need to invert the duality depending on which image plane we start with.) Also, we will not accept solutions differing in similarity transformations, only point-wise periodicity of the corresponding diagrams. Furthermore, we will treat the reciprocity condition independently, meaning that the fact that a duality and corresponding projection satisfy it does not mean the two diagrams are each other's solutions in the mechanical sense. In order to formalize this, let us start with polyhedrons $\mathcal{P}$, $\mathcal{D}$ and camera $C$ such that the two diagrams are $\Gamma_1=C(\mathcal{P})$ and $\Gamma_2=C(\mathcal{D})$. Let us denote the projection center of $C$ with $i$, the image plane of the projection with $j$ and the embedding of $\Gamma_1 \in PG(2)$ and $\Gamma_2 \in PG(2)$ into $PG(3)$ with $\bar{\mathcal{P}}$ and $\bar{\mathcal{D}}$ respectively (see Appendix \ref{sec:appa} or \ref{sec:appb}). It is important to note that $i$ and $\bar{\mathcal{P}}$ does not determine $\mathcal{P}$ uniquely, only $p_k=i+\lambda_k \bar{p}_k$ has to hold, for some $\lambda_k\in \mathbb{R}\setminus \{ 0 \}$ (where $p_k\in \mathcal{P}$ and $\bar{p}_k \in \bar{\mathcal{P}}$). After this set-up, we can finally address the reciprocity condition with 

\begin{theorem}\label{thm:period}
The reciprocity condition is equivalent with $D^{-T}D$ (the duality iterated twice, mapping points to points) being a central-axial collineation.
\end{theorem}
\begin{proof}
The equivalence will be proven with the usual bipartite argument.
\paragraph{Periodicity $\implies$ central axial collineation}
Since vectors $\vec{D}p_k$ represent the planes of $\mathcal{D}$ by definition, the dual of $\mathcal{D}$ consists of the set of points given by  $D^{-T}D p_k$, which is a projective transformation of $\mathcal{P}$. If we want to see the same image, we have to have $\Gamma_1=C(\mathcal{P}) = C(D^{-T}D\mathcal{P})$. Point-wise this means $D^{-T}D p_k=i+\lambda_k' \bar{p}_k$ has to hold, for all $p_k$ with some $\lambda_k' \in \mathbb{R}\setminus \{ 0 \}$. This means all lines $p_k=i+\lambda_k \bar{p}_k$ (all lines passing through $i$) are fixed under $D^{-T}D$. (Since the dual of the dual of $\mathcal{D}$ consists of planes $\vec{D}D^{-T}Dp_k=DD^{-T}\vec{D}p_k=(D^{-T}D)^{-T}\vec{D}p_k$, $\Gamma_2=C(\mathcal{D}) = C(D^{-T}D\mathcal{D})$ would imply the same thing).
\paragraph{Central axial collineation $\implies$ periodicity}
It is trivial to see how the choice of the collineation centre as the projection centre $i$ satisfies the periodicity of the images.
\end{proof}

An illustration of such system can be seen in Figure \ref{fig:kep1}, with the underlying coordinates are given in Appendix \ref{sec:appb}. Similarly as polarities are a subset of dualities, polarities are a subset of the dualities satisfying Theorem \ref{thm:period}, as the identity (any polarity iterated twice) is a central axial collineation.

\begin{figure}[h]
\includegraphics[width=0.5\textwidth]{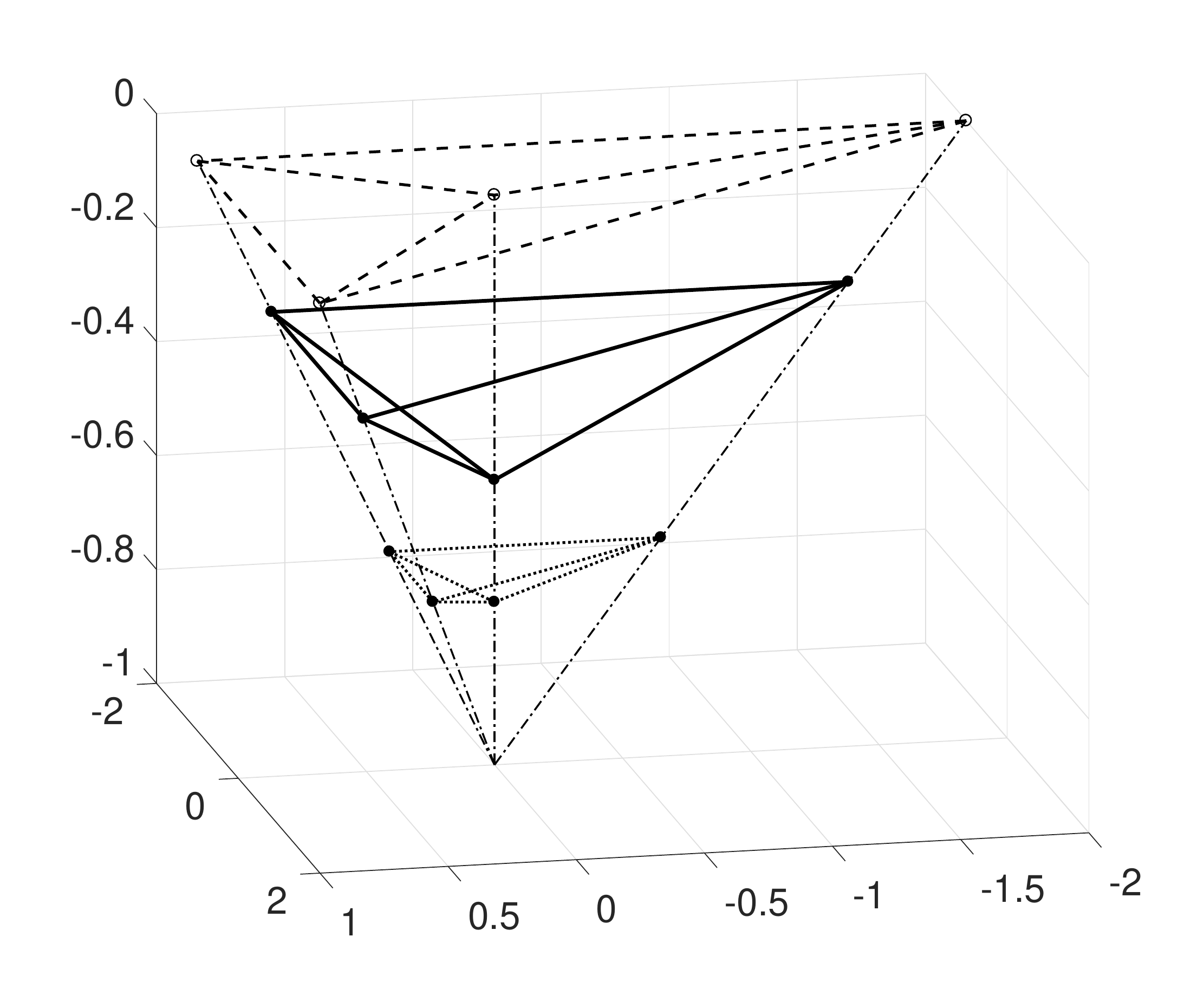}
\includegraphics[width=0.5\textwidth]{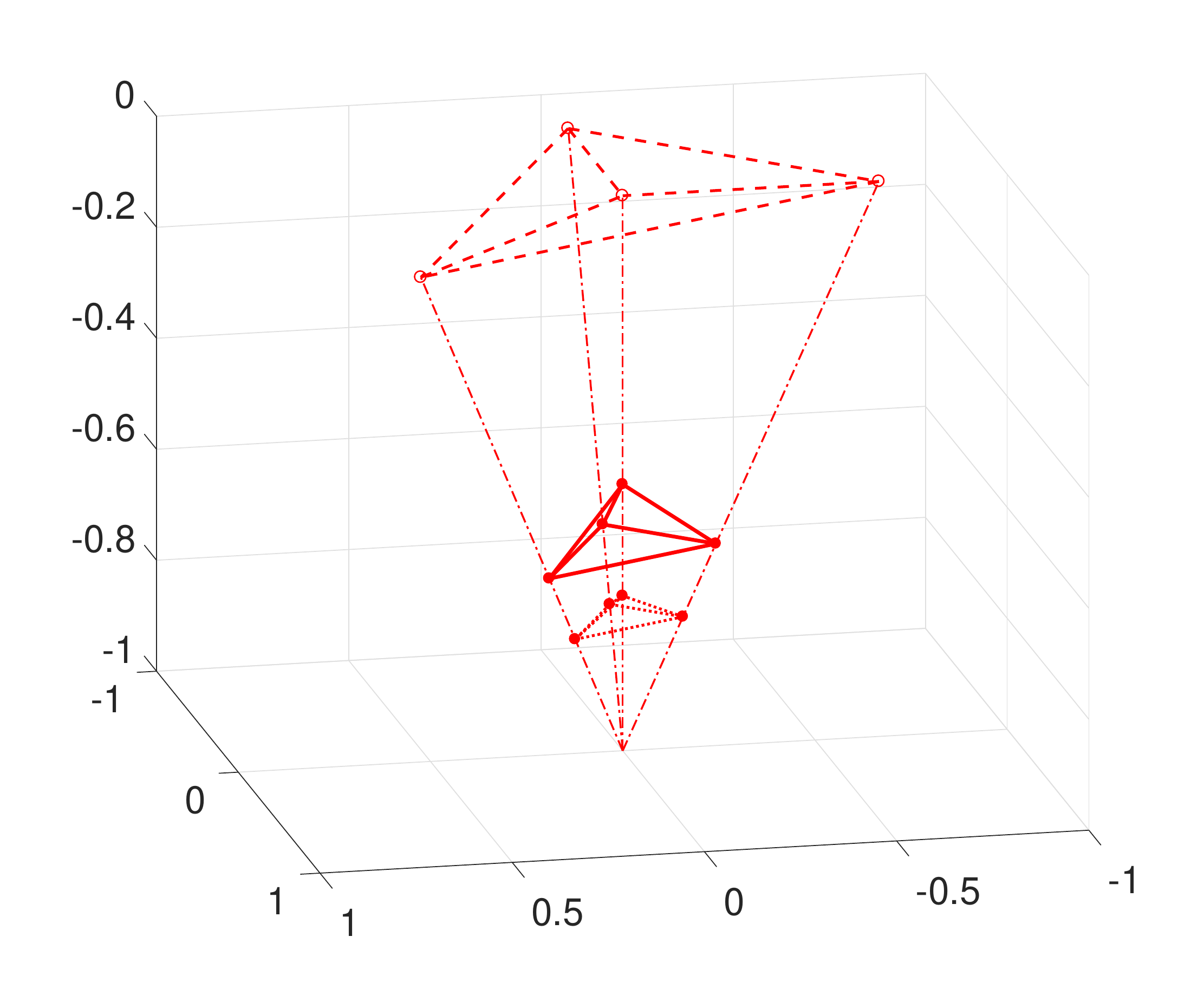}
\caption{Left: polyhedron $\mathcal{P}$ (solid lines), $D^{-T}D(\mathcal{P})$ (dotted lines) as well as their projecion. Right: polyhedron $\mathcal{D}$ (solid lines), $D^{-T}D(\mathcal{D})$ (dotted lines) as well as their projecion.}\label{fig:kep1}
\end{figure}


\begin{appendices}
\section{Maxwells method with homogenous coordinates}\label{sec:appa}
The tools Maxwell gave us are duality \begin{align*}
\vec{A}=
\begin{bmatrix}
1 & 0 & 0 & 0 \\
0 & 1 & 0 & 0 \\
0 & 0 & 0 & 1\\
0 & 0 & 1 & 0
\end{bmatrix}
\end{align*}
and projection \begin{align*}
C_{\infty}=
\begin{bmatrix}
1 & 0 & 0 & 0\\
0 & 1 & 0 & 0\\
0 & 0 & 0 & 1
\end{bmatrix}.
\end{align*}
The duality is often represented with a paraboloid of revolution, since the set of all points satisfying $x^T\vec{A}x=0$ form a paraboloid of revolution. Since the duality maps point $x$ to plane $\vec{A}x$, this means these points are incident with their duals, so called self-conjugate. Also, $\vec{A}x$ is a the tangent plane of the paraboloid at point $x$. The projection  $C_{\infty}$ is a parallel projection to a horizontal image plane, the centre of projection in the projective sense is point $i=(0,0,1,0)$.
Let us have the first diagram  $\Gamma_1 \in PG(2)$ given with its points as $u_0=(0,0,1)$, $u_1=(1,-1,1)$, $u_2=(1,2,1)$ and $u_3=(-2,-1,1)$. The image plane for our purposes will be $z=0$. We can embed $\Gamma_1$ into $PG(3)$ and batch up the vertices of the diagram for convenience by introducing

\begin{align}
\bar{\mathcal{P}}=
\left[
\begin{array}{c|c|c|c}
\bar{p}_0 & \bar{p}_1  & \bar{p}_2 & \bar{p}_3
\end{array}
\right]=\left[
\begin{array}{c|c|c|c}
0 & 1  & 1 & -2 \\
0 & -1 & 2 & -1 \\
0 & 0  & 0 & 0 \\
1 & 1  & 1 & 1
\end{array}
\right].\label{eq:barp}
\end{align} 
Then, we have to lift $\bar{\mathcal{P}}$ to form a (non-degenerate) polyhedron $\mathcal{P}$, which can be done element-wise, for instance with $p_k=\lambda_k\bar{p}_k+(1-\lambda_k)i$. As long as the result is non-degenerate, any combination of $\lambda_k$ is allowed, in this example we will use $\lambda_0=0.5$ and $\lambda_1=\lambda_2=\lambda_3=0.75$. With this we have
\begin{align}
\mathcal{P}=\left[
\begin{array}{c|c|c|c}
0 & 0.75  & 0.75 & -1.5 \\
0 & -0.75 & 1.5 & -0.75 \\
0.5 & 0.25  & 0.25 & 0.25 \\
0.5 & 0.75  & 0.75 & 0.75
\end{array}
\right] \sim
\left[
\begin{array}{c|c|c|c}
    0  &  1   &  1   & -2\\
    0  & -1   &  2   & -1\\
    1  &  1/3 &  1/3 &  1/3\\
    1  &  1   &  1   &  1
\end{array}
\right].\label{eq:Pset}
\end{align}
We can now take the dual of $\mathcal{P}$, keeping in mind that the set \begin{align}
\mathcal{F}:=\vec{A}\mathcal{P}=\left[
\begin{array}{c|c|c|c}
    0  &  1   &  1   & -2\\
    0  & -1   &  2   & -1\\
        1  &  1   &  1   &  1\\
        1  &  1/3 &  1/3 &  1/3
\end{array}
\right]
\end{align} is a set of planes. In order to get the vertices of the dual polyhedron $\mathcal{D}$ in this particular example, we will index the vertices of $\mathcal{D}$ such that vertex $d_k$ is incident with plane $f_l$ if $k\neq l$ ($\scal{d_k}{f_l}=0$ and $\scal{d_k}{f_k} \neq 0$). Such system can be efficiently solved by finding the (left) null-spaces of the matrices formed by different subsets of columns of $\mathcal{F}$. We arrive at
\begin{align}
\mathcal{D}=\left[
\begin{array}{c|c|c|c}
  0 &  -2/3    & 0 &  2/3 \\
  0 & 2/3 &  -2/3 & 0 \\
   -1/3 &  -1 &   -1 &   -1 \\
    1 & 1 & 1 & 1
\end{array}
\right]
\end{align} which can be projected down to the dual diagram
\begin{align}
\Gamma_2=C_{\infty}\mathcal{D}=\left[
\begin{array}{c|c|c|c}
  0 &  -2/3    & 0 &  2/3 \\
  0 & 2/3 &  -2/3 & 0 \\
    1 & 1 & 1 & 1
\end{array}
\right].
\end{align}
The illustration of this process is presented in Figure \ref{fig:mkep}.

\section{Example undoing of an affinity}\label{sec:appc}
Consider polyhedron $\mathcal{P}$ given in \eqref{eq:Pset} and affine transformation \begin{align}
U=\begin{bmatrix}
2/4 & 0 & 0 & 0\\
0 & 3/4 & 0 & 0\\
0 & 0 & 2 & -1 \\
0 & 0 & 0 & 1
\end{bmatrix}.
\end{align}
The transformation $U^{-1}$ maps $\mathcal{P}$ to
\begin{align}
U^{-1}\mathcal{P}=
\left[
\begin{array}{c|c|c|c}
    0  &  2   &  2   & -4\\
    0  &  -4/3 &  8/3 &  -4/3\\
    1  &  2/3   &  2/3   &  2/3\\
    1  &  1   &  1   &  1
\end{array}
\right].
\end{align} 
The original image plane $z=0$ of Appendix \ref{sec:appa} is mapped into $z=-1$. The ratio of  distortion $U^{-1}$ causes on this new image plane is  $(2/4)/(3/4)=2/3$. In order to undo this, we need an image plane with a normal vector parallel to $(-2/3,0,\sqrt{1-4/9})$. One such plane is given with homogenous coordinates $(-2/3,0,\sqrt{1-4/9},3)$ and the arising system of images is drawn in Figure \ref{fig:affin}. What is important to note is how the image in plane $(-2/3,0,\sqrt{1-4/9},3)$ is a uniformly scaled up version of the image in plane $z=0$ of Appendix \ref{sec:appa}.
\section{Reciprocal Example}\label{sec:appb}
Here a numerical example of a reciprocal system is provided. Consider duality 
\begin{align}
\vec{D}=\begin{bmatrix}
    0 &    -1 &     0 &      0\\
    1 &     0 &     0 &      0\\
    0 &     0 &  -0.75 &  -1.75\\
    0 &     0 &   0.25 &  -0.75\\
\end{bmatrix}
\end{align}
and camera matrix
\begin{align}
C=
\begin{bmatrix}
1 & 0 & 0 & 0\\
0 & 1 & 0 & 0\\
0 & 0 & 1 & 1
\end{bmatrix}.
\end{align}

Since we want to use a single projection and not the decomposition given in the proof of Theorem \ref{thm:jo}, we have to show that the images it gives are correct: that is they can be mapped to the known solutions through similarity transformations. For this purpose, the following lemma will be of use:
\begin{lemma}\label{lem:relative}
Given images $C_1(\mathcal{P})$ and $C_2(\mathcal{P})$ of polyhedron $\mathcal{P}$, the perspective transformation mapping the first to the second is given by $C_2C_1^+$ where $C_1^+$ is the pseudo-inverse of $C_1$.
\end{lemma}
The proof of this can be found in \cite{xu2013epipolar}. Now, using the decomposition $\vec{D}=Id\vec{Id}D$, and the images introduced in the proof of Lemma \ref{lem:inf} we can have
\begin{align}
CD^{-1}C_0^+=\begin{bmatrix}
     0  &   1  &   0\\
    -1  &   0  &   0\\
     0  &   0  &  -1
\end{bmatrix} \text{ and } CIdC_{\pi/2}^+=\begin{bmatrix}
     0 &    1  &   0\\
     -1 &    0  &   0\\
     0 &    0  &   1
\end{bmatrix}
\end{align}
both of which is in $\mathcal{S}$, meaning we have good solutions. 
As for the periodicity condition, we have 
\begin{align}
D^{-T}D=\begin{bmatrix}
   -1 & 0 & 0 & 0\\
    0 & -1 & 0 & 0\\
    0 &  0  & 0.5 & 1.5\\
    0 & 0  & -1.5 & -2.5
\end{bmatrix}
\end{align}
with $-1$ as a 4-fold eigenvalue. It has only 3 eigenvectors, as the center is on the fixed plane, thus the transformation is an elation: the centre $i=(0,0,-1,1)$ is in the fixed plane $z=-1$. As for the polyhedrons, the process given in Appendix \ref{sec:appa} can be repeated with the different projection centre and duality. The whole system can be seen in Figure \ref{fig:kep1}.

\end{appendices}
\bibliographystyle{unsrt} 
\bibliography{nempolarbib}

\begin{thebibliography}{1}

\bibitem{pottmann2001computational}
H.~Pottmann and J.~Wallner.
\newblock {\em Computational Line Geometry}.
\newblock Mathematics and Visualization. Springer Berlin Heidelberg, 2001.

\bibitem{hartley2003multiple}
R.~Hartley and A.~Zisserman.
\newblock {\em Multiple View Geometry in Computer Vision}.
\newblock Cambridge books online. Cambridge University Press, 2003.

\bibitem{KONSTANTATOU2018272}
Marina Konstantatou, Pierluigi D'Acunto, and Allan McRobie.
\newblock Polarities in structural analysis and design: n-dimensional graphic
  statics and structural transformations.
\newblock {\em International Journal of Solids and Structures}, 152-153:272 --
  293, 2018.

\bibitem{carrell2017groups}
J.B. Carrell.
\newblock {\em Groups, Matrices, and Vector Spaces: A Group Theoretic Approach
  to Linear Algebra}.
\newblock Springer New York, 2017.

\bibitem{xu2013epipolar}
G.~Xu and Z.~Zhang.
\newblock {\em Epipolar Geometry in Stereo, Motion and Object Recognition: A
  Unified Approach}.
\newblock Computational Imaging and Vision. Springer Netherlands, 2013.

\end{thebibliography}

\end{document}